\documentclass[prd,aps,amsfonts,showpacs,longbibliography,notitlepage,twocolumn,groupedaddress]{revtex4-1}
\usepackage[dvipsnames]{xcolor}
\usepackage{multirow}
\usepackage{tabularx}

\usepackage{amsmath,amssymb,amsfonts,mathtools,dsfont,relsize}
\usepackage{graphicx}

\usepackage[colorlinks=true, urlcolor=violet, linkcolor=blue, citecolor=red, hyperindex=true, linktocpage=true]{hyperref}

\usepackage{amsthm}
\usepackage{qtree}
\newtheorem{thm}{Theorem}

\newtheorem{lem}[thm]{Lemma}
\newtheorem{prop}[thm]{Proposition}

\usepackage{soul} % what is this?
\usepackage{bm}

\usepackage{enumerate}
\usepackage{stmaryrd} % what is this??

\renewcommand{\thesubsection}{\thesection.\arabic{subsection}}

\makeatletter
\renewcommand{\p@subsection}{}
\renewcommand{\p@subsubsection}{}
\makeatother

\DeclarePairedDelimiter{\ket}{\lvert}{\rangle}
\DeclarePairedDelimiter{\bra}{\langle}{\rvert}
\newcommand{\ii}[0]{\mathrm{i}}
\usepackage{qcircuit}

\newcommand{\lr}[1]{\left( #1\right)}
\newcommand{\mlr}[1]{\left[ #1\right]}

\newcommand{\alr}[1]{\left\langle #1\right\rangle}
\newcommand{\norm}[1]{\left\lVert#1\right\rVert}
\newcommand{\abs}[1]{\left\lvert#1\right\rvert}

\newcommand{\ee}{\mathrm{e}}
\newcommand{\dd}{\mathrm{d}}

\newcommand{\kpsi}{\ket{\psi}}
\newcommand{\kpsin}[1]{\ket{\psi_{#1}}}

\newcommand{\bz}{{\mathbf{z}}}
\newcommand{\bs}{{\mathbf{s}}}

\makeatletter 
    
\renewcommand\onecolumngrid{% <<<<<<
\do@columngrid{one}{\@ne}%
\def\set@footnotewidth{\onecolumngrid}% <<<<<<<<<<<<<<<<
\def\footnoterule{\kern-6pt\hrule width 1.5in\kern6pt}%
}

\renewcommand\twocolumngrid{% <<<<<<
        \def\footnoterule{% restore rule
        \dimen@\skip\footins\divide\dimen@\thr@@
        \kern-\dimen@\hrule width.5in\kern\dimen@}
        \do@columngrid{mlt}{\tw@}
}%

\makeatother

\newcolumntype{Y}{>{\centering\arraybackslash}X}
\newcolumntype{M}[1]{>{\centering\arraybackslash}m{#1}}

\newcommand{\comment}[1]{}

\newcommand{\revise}[1]{#1}

\begin{document}
\title{Eigenstate localization in a many-body quantum system}

\author{Chao Yin}
\affiliation{Department of Physics and Center for Theory of Quantum Matter, University of Colorado, Boulder, CO 80309, USA}

\author{Rahul Nandkishore}
\affiliation{Department of Physics and Center for Theory of Quantum Matter, University of Colorado, Boulder, CO 80309, USA}

%science has one corresponding author?
\author{Andrew Lucas}
\email{andrew.j.lucas@colorado.edu}
\affiliation{Department of Physics and Center for Theory of Quantum Matter, University of Colorado, Boulder, CO 80309, USA}

%NEW ABSTRACT, FOR PRL?
\begin{abstract}
We prove the existence of extensive many-body Hamiltonians with few-body interactions and a many-body mobility edge: all eigenstates below a nonzero energy density are localized in an exponentially small fraction of ``energetically allowed configurations" within Hilbert space.  Our construction is based on quantum perturbations to a classical low-density parity check code.  In principle, it is possible to detect this eigenstate localization   by measuring few-body correlation functions in efficiently preparable mixed states.

%%%%% ORIGINAL ABSTRACT
%Generic many-body interacting systems obey the laws of statistical mechanics, including the ergodic hypothesis: the system explores all microstate configurations consistent with conservation laws.  In quantum systems, ergodicity implies that many-body eigenstates are delocalized among allowed configurations.  Here, we nonperturbatively prove  the existence of a many-body quantum system with an extensive Hamiltonian and few-qubit interactions, for which all eigenstates below a critical energy density are many-body localized.  Surprisingly, our construction neither relies on strong randomness nor on low spatial dimensionality.  The Hamiltonian is a classical constant-rate error correcting code with generic quantum few-body interactions. We conclusively verify the longstanding conjecture that eigenstate many-body localization is possible, and that it can be robust to generic perturbations.  Our work reveals a profound connection between quantum statistical mechanics and classical error correcting codes,  with the latter protecting thermodynamically large, but finite, quantum systems which do not thermalize in infinite time.

\end{abstract}

\date{\today}

 \maketitle

\emph{Introduction}.--- Statistical mechanics is the framework that encapsulates how complex many-body systems can be described by simple emergent models, such as hydrodynamics.  It assumes ergodicity: a system explores all ``energetically allowed configurations" with equal probability during its time evolution.

There is an old paradox associated with ergodicity in many-body quantum systems.  States obey the Schr\"odinger equation: setting Planck's constant $\hbar=1$, \begin{equation}\label{eq:schrodinger}
    \frac{\mathrm{d}}{\mathrm{d}t} |\psi(t)\rangle = -\mathrm{i}H|\psi(t)\rangle.
\end{equation}
We formally solve (\ref{eq:schrodinger}) by diagonalizing the matrix $H$. If we are handed an eigenvector $|\psi_0\rangle$ obeying $H|\psi_0\rangle = E|\psi_0\rangle $, then its time evolution is trivial:  $|\psi_0(t)\rangle = \mathrm{e}^{-\mathrm{i}Et}|\psi_0\rangle$.  Therefore, no physical observable evolves in time.  This seems to contradict ergodicity outright. This paradox is resolved by the eigenstate thermalization hypothesis (ETH), which states that $|\psi_0\rangle$ itself must appear thermal: for any few-body observable $A$, \cite{deutsch,srednicki}
\begin{equation}
    \langle \psi_0 | A |\psi_0\rangle \approx \frac{\mathrm{tr}\left(\mathrm{e}^{-\beta H} A\right)}{\mathrm{tr}\left(\mathrm{e}^{-\beta H}\right)}. \label{eq:ergodic}
\end{equation}
Here $\beta$ is the inverse temperature associated with the energy $E$ of the eigenstate; the right hand side of (\ref{eq:ergodic}) gives a precise meaning to sampling over ``allowed configurations".  Extensive numerical simulations are consistent with ETH in a wide range of quantum many-body systems \cite{mblarcmp, mblrmp,Deutsch_2018}. 
 
Given the ubiquity of ETH, it is tempting to find a many-body quantum system where the ETH, and in turn the theory of statistical mechanics, fails. Since (\ref{eq:ergodic}) is a statement about \emph{eigenstates} of a system, which are only well-defined if the number of particles $N$ is kept finite, ETH makes sense if we \emph{first consider $t\rightarrow \infty$, then $N\rightarrow \infty$.}   Keeping this order of limits in mind,  it is well known that the ETH can fail at \emph{fine tuned points} in the space of all Hamiltonians.  Systems can be integrable \cite{Baxter} and possess extensive conserved quantities.  They may also have quantum scars, or special eigenstates which violate ETH even when typical eigenstates do obey ETH \cite{scar_PXP18,scar_exact18,scar_exact18_1,scar_rev21,scarsarcmp}. The Hilbert space can also fragment (shatter) into disconnected subsets with no matrix elements allowing for quantum dynamics between these subsets  \cite{KHN, Sala2020, MPNRB,moudgalya2022quantum,algebra_Motrunich22}; this structure can even be robust to exponentially long but finite times when certain emergent symmetries are present  \cite{KHN, SHN, SNH}; see also \cite{word_Lake23}.  Yet all of these mechanisms are believed, in many-body systems, to be non-robust to generic perturbations in the $t\rightarrow \infty$ limit at finite $N$.

In single-particle quantum mechanics, there is a way to obtain robust ergodicity breaking in the limit $t\rightarrow \infty$, known as \emph{Anderson localization} \cite{anderson}.  When a particle hops on a lattice in the presence of disorder, the eigenvalue equation $H|\psi_0\rangle = E|\psi_0\rangle$ can become unsolvable unless $|\psi_0\rangle$ is localized if the energy $E$ is \emph{off resonance} with the local energy scales in $H$ away from an isolated region of space.    Clearly, this localized eigenstate violates ETH in the large system limit: there are many other configurations of similar energy that are inaccessible.

We now turn to the many-body setting.  Consider $N$ interacting qubits, whose Hilbert space is spanned by bitstrings $|\mathbf{x}\rangle$ where $\mathbf{x}\in\mathbb{F}_2^N$, where $\mathbb{F}_2=\lbrace 0,1\rbrace$. Can eigenstates localize in the space of bitstrings $|\mathbf{x}\rangle$? This problem is substantially harder than single-particle Anderson localization: in physical systems, the Hamiltonian $H$ takes the form \begin{equation}\label{eq:Hdecompose}
    H = \sum_{S\subset \lbrace 1,\ldots, N\rbrace : |S| \le q} H_S \otimes I_{S^{\mathrm{c}}},
\end{equation}where $H_S$ acts only on a subset $S$ of at most $q$ degrees of freedom and $I_{S^{\mathrm{c}}}$ is the identity matrix on the remainder. We further demand that each degree of freedom interacts with finitely-many others, so that adjusting any one qubit can only change the energy by an $N$-independent O(1) amount.  Notice that for any given bit string $|\mathbf{x}\rangle$, there are at least $\mathrm{O}(N)$ bitstrings $\mathbf{x}^\prime$ for which $\langle \mathbf{x}|H|\mathbf{x}^\prime\rangle \ne 0$.  As a consequence, it is far from clear that the ``no resonance" condition responsible for Anderson localization can localize a many-body eigenstate.  

Nevertheless, it has been conjectured for almost 70 years  \cite{anderson, AGKL, GMP, Basko_2006, huse2007} that \emph{many-body localization} (MBL) is possible, with disordered quantum spin chains believed to be the most likely setting. However, extensive work \cite{liom,lbits,liom2,drewpotter,serbyn,serbyn2,morningstar,MBL_challenge21,MBL_DEROECK23}  (see \cite{mblarcmp, mblrmp} for reviews), has been unable to give a complete proof of the existence of MBL. There is a long history  \cite{AGKL, tarzia,RRG19,RRGrev21} of attempting to model MBL by ``cartoons" \cite{levvidmar}, in particular single-particle Anderson localization on random graphs \cite{abouchacra,Mirlin_analytic91,Bethe12,Bethe14,RRG16,RRG_finitesize16,RRG19,RRGrev21,RRG23}.  This is not carefully justified on mathematical grounds; the many-body interaction graph has (at least) $\mathrm{O}(N)$ connectivity, many loops, and strong correlations between disorder at different points on the graph.  A rigorous derivation of eigenstate localization must explicitly address all of these challenges. The most formal arguments \cite{Imbrie_2016} rely on a plausible assumption (`no strong level attraction' in the many-body spectrum), yet the ultimate conclusion of MBL has recently been challenged \cite{vidmar, selspolkovnikov}. There is an abundance of evidence, both theoretical \cite{deRoeckHuveneers, GopalakrishnanHuse,preth_MBL23,MBL_resonance_Huse23} and experimental \cite{bloch, Choi_2016, DeMarco}, for MBL as (at least) a prethermal phenomenon, persisting over non-perturbatively large (but finite, in the thermodynamic limit) times. 

Here, we present a family of many-body quantum systems in which every low-energy eigenstate is proved to be localized, and settle the longstanding conjecture that such eigenstate localization is possible.  Amusingly, we do not study disordered spin chains; instead, we use good classical error-correcting codes \cite{Sipser_1996} as the basis for a many-body quantum system with localized eigenstates.   Our construction is related to previous work \cite{baldwin,baldwin2,leschke,Winer:2022ciz}, which has heuristic arguments for similar eigenstate localization in a genuine many-body problem (which we are able to rigorously prove). The authors of \cite{baldwin} described said eigenstate localization as `non-ergodic but not MBL.' 
We do not agree with this distinction. Eigenstate localization in our model occurs within a connected region of Hilbert space,  is robust to perturbations, has deep mathematical analogies to single-particle Anderson localization in three dimensions, and involves a many-body mobility edge at non-zero energy density, just as in the original works on MBL \cite{GMP, Basko_2006}; see also \cite{warzel}. As such, we believe it makes sense to call it MBL, although of a different kind than postulated in one dimension.

\emph{Classical error correcting codes}.--- To explain our construction, we must first review the theory of classical binary linear error correcting codes, which store $K$ logical bits in $N>K$ physical bits.  Of the possible $2^N$ physical bitstrings $\mathbf{x}\in\mathbb{F}_2^N$, $2^K$ of them correspond to logical \emph{codewords}.  The code distance $D$ is defined to be the smallest nonzero Hamming weight (number of 1s in the bitstring) of a codeword $\mathbf{z}$, %\rahul{shouldn't it be Hamming distance to the nearest other codeword?} \andy{for linear code this is an equivalent criterion, since if $\mathbf{0}$, $\mathbf{x}_1$ and $\mathbf{x}_2$ are codewords so is $\mathbf{x}_1-\mathbf{x}_2$} \rahul{ok}
  denoted as $|\mathbf{z}|$.  In a linear code, the codewords are the right null vectors of the parity check matrix $\mathsf{H} \in \mathbb{F}_2^{M\times N}$\revise{, where $M$ is the number of parity checks}; the right null space thus has dimension $K$.  Notice that one codeword is guaranteed to be $\mathbf{x}=\mathbf{0}$. Of interest are low-density parity check (LDPC) codes \cite{Sipser_1996}, where $\mathsf{H}$ is sparse: each row and column has at most $q=\mathrm{O}(1)$ 1s.  
%We will find it convenient to denote the entries of $\widetilde{H}$ as $\widetilde{H}_{Ci}$, which determines whether physical bit $i$ is contained within parity check $C$.  Remarkably, many LDPC codes can achieve $K,D=\mathrm{O}(N)$.   

A very useful type of LDPC code called a c3LTC has recently been constructed \cite{Panteleev_2021,dinur2022good,dinur2022LTC,lin2022}, for which $q$ is O(1), $D=\mathrm{O}(N)$, $K=\mathrm{O}(N)$, and we have a valuable property known as \emph{local testability} (LT), which implies that the parity check matrix $\mathsf{H}$ has $\mathrm{O}(N)$ left null vectors (redundancies among parity checks), such that any configuration violating few parity checks is close to a codeword.  More precisely, for any bitstring $\mathbf{x}\in\mathbb{F}_2^N$, the number of violated parity checks obeys \begin{equation}
    |\mathsf{H}\mathbf{x}| \ge \alpha \min_{\text{codeword }\mathbf{z}} |\mathbf{x}-\mathbf{z}| \label{eq:confinement}
\end{equation}
 for some O(1) $\alpha>0$. (\ref{eq:confinement}) is called \emph{linear soundness} in the literature and is helpful to us.  In particular, \revise{ linear soundness implies that any bitstring far from all codewords necessarily flips an $\mathrm{O}(1)$ fraction of parity checks, and is a finite energy density state. This in turn implies LT.}  More general LDPC codes have a property analogous to (\ref{eq:confinement}) that holds locally near low-energy states \cite{Sipser_1996}, and this also leads to eigenstate localization with a few complications: see the Supplementary Material (SM) \footnote{SM also contains formal statements and proofs.} for details.  (\ref{eq:confinement}) is impossible in a code which is geometrically local in finite spatial dimension $d$: nucleating a bubble of radius $R$ inside of which the configuration corresponds to codeword $\mathbf{z}$, outside of which there is codeword $\mathbf{0}$, flips $\mathrm{O}(R^d)$ bits, violating $\mathrm{O}(R^{d-1})$ parity checks.
  %The facts stated in this paragraph are all that we will need to prove the existence of MBL.

Given parity check matrix $\mathsf{H}$, we can define a classical $q$-local Hamiltonian (with $\le q$-body interactions): \begin{equation}
    H_0 = \frac{1}{2} \sum_{\text{parity check }C} \left[1- \prod_{i \in C} Z_i\right] = \sum_C P_C.  \label{eq:H0}
\end{equation}
For later convenience, the parameters $1-2x_i = Z_i \in \lbrace \pm 1 \rbrace$, rather than $\mathbb{F}_2$; we also defined shorthand $P_C$.  $i\in C$ means $\mathsf{H}_{Ci}=1$.  Notice that $H_0= |\mathsf{H}\mathbf{x}| $.

It is illustrative to pause and study a simple example. If our parity checks $C\in \lbrace 1,\ldots, n-1\rbrace$ while $i\in \lbrace 1,\ldots, n\rbrace$, we can consider parity check matrix \begin{equation}
    \mathsf{H}_{Ci} = \left\lbrace\begin{array}{ll} 1 &\ C=i \text{ or } i-1 \\ 0 &\ \text{otherwise} \end{array}\right.,
\end{equation}
which leads to the 1d Ising model: \begin{equation}
    H_0 = \sum_{n=1}^{N-1} \frac{1-Z_i Z_{i+1}}{2},
\end{equation}
known in information theory as the repetition code.  The parity checks are then simply the ferromagnetic interactions that prefer to align nearby spins, while the codewords are the states where all $Z_i = +1$ (codeword $0\cdots 0$) or all $Z_i = -1$ (codeword $1\cdots 1$).   We can add redundant parity checks by replacing the 1d Ising model with the 2d Ising model.  This does not change the codewords, but we do gain a weaker form of LT, in which all states with $\ll \sqrt{N}$ violated parity checks are close to a codeword and easily decodable. As is well-known, this is sufficient to cause a (ferromagnetic) thermal phase transition: (almost) all low energy states clustered near codewords.

%With the Ising analogy in mind, we now state more formally the properties of the c3LTCs.  For sufficiently large $N$, there exist O(1) (positive $N$-indpendent) constants $k_0$, $d_0$ and $m_0$ such that we can build an LDPC code with at most $q$-local parity checks, and where each bit is contained in at most $q$ parity checks, while maintaining $K\ge k_0 N$, $D\ge d_0 N$, $M \le m_0 N$.  That $m_0$ is finite guarantees that Hamiltonian $H_0$ is extensive and has physically acceptable thermodynamics; finite $k_0$ ensures exponentially many codewords.  LT implies that 

\emph{Our model}.--- We are now ready to return to many-body quantum mechanics. We can interpret Hamiltonian $H_0$ in (\ref{eq:H0}) as a ``classical Hamiltonian" on this quantum Hilbert space, where $Z$s simply represent Pauli matrices.  We say that $H_0$ is classical because it is trivial to diagonalize: its eigenvectors are $|\mathbf{s}\rangle$ and eigenvalues are the number of violated parity checks in the bitstring $\mathbf{s}$.  
 $H_0$ has a very large symmetry group $\mathbb{Z}_2^K$ consisting of all operators \begin{equation}
    X_{\mathbf{z}} = \prod_{i : z_i=1 \text{ in codeword }\mathbf{z}} X_i.
\end{equation}
These operators correspond to shifting the state of the classical code by codeword $\mathbf{z}$, which by definition does not modify any parity check.  Hence, at the quantum level, $[H_0,X_{\mathbf{z}}] = 0$.

Upon choosing a c3LTC, we introduce the quantum Hamiltonian \begin{equation}
    H = H_0 + H_{\mathrm{SB}} +V + H_{\mathrm{L}}, \label{eq:H}
\end{equation}where $H_0$ is given by (\ref{eq:H0}).  The remaining three terms are as follows.  Firstly, we introduce the symmetry-breaking \begin{equation}
    H_{\mathrm{SB}} = \sum_C \sum_{i\in C} J_{Ci} Z_i P_C, \;\;\; |J_{Ci}| =  \frac{1}{2 q},  \label{eq:HSB}
\end{equation}
where the restriction on $J_{Ci}$ is chosen such that the analogue of (\ref{eq:confinement}) continues to hold up to $\alpha \rightarrow \alpha/2$, and we observe that $H_{\mathrm{SB}}$ does not modify the $q$-locality of $H$.  The $J_{Ci}$ \revise{with arbitrary signs} are chosen to not be perturbatively small, so that $H$ is not close to $H_0$, and are also chosen to break all of the $\mathbb{Z}_2^K$ symmetries of the problem.  This latter step is important as eigenstates of $H_0$ necessarily transform in irreducible representations of any exact symmetries, which can delocalize them.  $V$ is a generic perturbation which can be decomposed as in (\ref{eq:Hdecompose}); we assume that it is $\Delta^\prime$-local (for O(1) $\Delta^\prime$), and that for each site $i$, \begin{equation}
    \left\lVert \sum_{S: i\in S} V_S \right\rVert \le \epsilon . \label{eq:Vbound}
\end{equation}
$V$ breaks the solvability of $H$: eigenstates are now linear combinations of exponentially many bitstrings $|\mathbf{s}\rangle$.
Here $\epsilon \ll 1$ will be perturbatively small.  Lastly,
\begin{equation}
    H_{\rm L} = \frac{\epsilon}{\sqrt{N}} \sum_i h_i Z_i,
\end{equation}
where $h_i$ are independent and identically distributed zero-mean, unit-variance Gaussian random variables.

The main result of this paper is that for sufficiently low energy density and sufficiently small O(1) $\epsilon$, given a generic Hamiltonian of the form (\ref{eq:H}), almost surely in the thermodynamic limit $N\rightarrow \infty$, \revise{ all eigenstates of $H$ with energy $E\le \epsilon_* N$} are many-body localized near a single codeword.  \revise{We show in the SM that there are exponentially many such localized eigenstates.}  We expect that this is a genuine many-body mobility edge \cite{Basko_2006}, in contrast to ``full MBL" \cite{liom, lbits}, although we have not proved that high energy eigenstates must be delocalized. Since our construction is insensitive to $V$, so long as it obeys (\ref{eq:Vbound}), localization is robust to perturbations. A formal statement and proof of these claims are in the SM. 

\emph{Detectability.}--- As emphasized in the introduction, localization is inherently a question about finite $N$ systems in the $t\rightarrow\infty$ limit.  Nevertheless, it is helpful to ask whether eigenstate localization would have any ``experimental consequences".  It is a reasonable postulate that an experimentalist can neither prepare pure states, let alone eigenstates, and moreover can only measure few-body observables for sufficiently large $N$.  Given such restrictions, let us now show that localization is, in principle, detectable.  

Assume that we are handed a localized eigenstate $|\psi_0\rangle$,  trapped near a single codeword $\mathbf{z}$.  Hence, 
\begin{equation}
    \left| \sum_{i=1}^N (-1)^{z_i} \langle \psi_0|Z_i|\psi_0\rangle \right| > N(1-\alpha d_0),
\end{equation}
which is $\mathrm{O}(N)$ larger than expected in a thermal ensemble. Hence, for a finite fraction of qubits $i$ and any low energy eigenstate $|\psi_0\rangle$, $\langle \psi_0 |Z_i|\psi_0\rangle$ fails to obey the eigenstate thermalization hypothesis (\ref{eq:ergodic}). Of course, an experimentalist cannot directly prepare $|\psi_0\rangle$, so we further show in the SM that given arbitrary initial $|\varphi\rangle$ supported on bitstrings $\mathbf{x}$ sufficiently close to codeword $\mathbf{z}$: $|\mathbf{x}-\mathbf{z}| \le \theta N$, for some O(1) $\theta$ and sufficiently small $\epsilon$, the time-evolved state $\mathrm{e}^{-\mathrm{i}Ht}|\varphi\rangle$ is trapped near the codeword for all $t$. By linearity of quantum mechanics, this conclusion is unchanged if handed a mixed state containing multiple $|\varphi\rangle$ trapped near a codeword.

This is a clear violation of ergodicity, even if we first take $t\rightarrow\infty$ before $N\rightarrow \infty$.  In either classical \cite{Montanari_2006,Hong:2024vlr} or semiclassical \cite{coleman} (`false vacuum') analyses of the problem, we would expect that the state could escape away from a single codeword in a time $\exp[\mathrm{O}(N)]$.   The fact that this escape can never occur is a clear consequence of eigenstate localization, and is mathematically analogous \cite{maybodorova,galitski} to a single quantum particle remaining trapped near a deep potential minimum for all time in a three-dimensional metal with a mobility edge.

\emph{Proof sketch}.--- We now summarize how we prove MBL. Our proof is surprisingly short, and is related to established techniques for Anderson localization.   We consider the most generic possible eigenstate $H|\psi_0\rangle = E|\psi_0\rangle$ with $E \le \epsilon_*N$. The first step is to show that $|\psi\rangle$ is localized near codewords (but maybe more than one). This follows directly from linear soundness (\ref{eq:confinement}):  all states far from codewords have $H_0 \gg \epsilon_* N$ and are off resonance: see Figure \ref{fig:wells}a.  More precisely, we can decompose $|\psi_0\rangle$ into a convenient basis: 
\begin{equation}\label{eq:psi=cn}
    |\psi_0\rangle = \sum_{\text{codeword }\mathbf{z}}\sum_{n=0}^{N_*} c_{\mathbf{z}n}|\mathbf{z}n\rangle + \sum_{n=N_*+1}^\infty c_n |n\rangle,
\end{equation}
where we define $N_* = \mathrm{O}(N)$ to be a cutoff between low/high energy states, $|\mathbf{z}n\rangle$ to be a sum over $Z$-basis states (bitstrings) obeying $H_0|\mathbf{z}n\rangle = n|\mathbf{z}n\rangle$ \emph{and} which are close to codeword $|\mathbf{z}\rangle = |\mathbf{z}0\rangle$, and $|n\rangle$ obeys $H_0|n\rangle = |n\rangle$.  In practice, it is useful to group a finite fraction of $n$s together into bunches $\tilde n$ such that the perturbation $V$ only couples $\tilde n$ to $\tilde n\pm 1$.  Importantly, due to linear soundness, this decomposition is unique; all low energy states with $n<N_*$ in (\ref{eq:psi=cn}) are close to a single $\mathbf{z}$.  Taking the inner product of the eigenvalue equation with $\langle n|$ and $\langle \mathbf{z}n|$, we obtain a collection of discretized one-dimensional Schr\"odinger equations in $H_0$-space, which are illustrated in  Figure \ref{fig:wells}b.  These one-dimensional lines emanate from each codeword and join at $n=N_*$.  It is straightforward to show that $|\psi_0\rangle$ is trapped at $n \lesssim E$, with exponentially suppressed tails at $n\sim N_*$.  Notice that linear soundness (\ref{eq:confinement}) is crucial: the wave functions are isolated within codewords and can only be joined through very high energy and off-resonant states.  A heuristic discussion of similar ideas is in \cite{baldwin,Altshuler_2010}.

\begin{figure}
    \centering
    \includegraphics[width=\columnwidth]{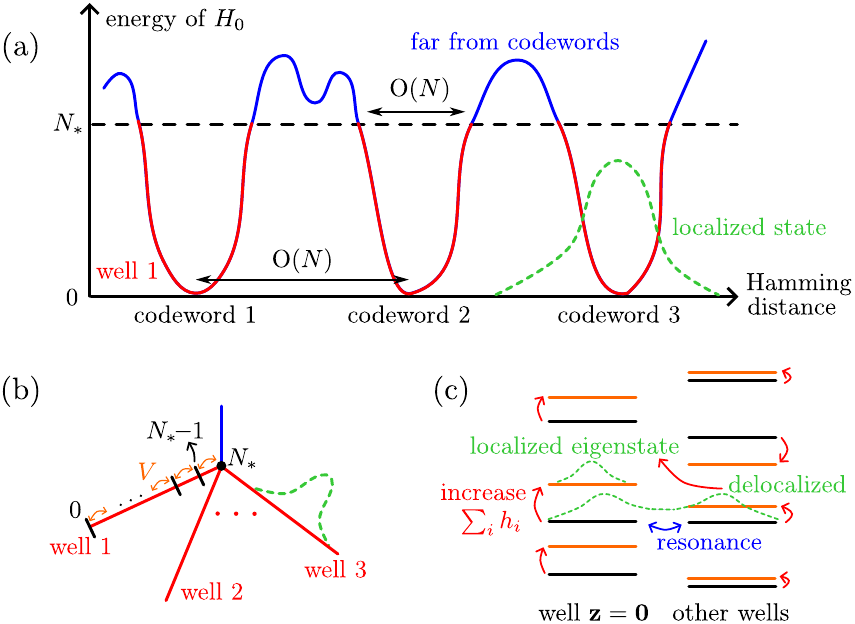}
    \caption{\textbf{Sketch of proof.} (a) The energy landscape of the c3LTC has deep wells near each codeword (ground state of $H_0$), with $\mathrm{O}(N)$ bit flips between each codeword and $\mathrm{O}(N)$ energy penalty to be far from all codewords.  (b) The many-body eigenvalue problem reduces to a collection of coupled one-dimensional quantum walks, where the ``emergent" dimension counts the number of flipped parity checks of $H_0$. Low-energy eigenstates are strongly localized close to codewords.  (c) A tiny amount of disorder in $H_{\mathrm{L}}$ breaks any accidental resonances between the projected Hamiltonians $H_{\mathbf{z}}$ near each codeword, localizing low energy density eigenstates of $H$ near a single codeword. }
    \label{fig:wells}
\end{figure}

% However, this is not enough to rule out that $|\psi\rangle$ has support in many codewords.  To see what could go wrong, it is instructive to first set $H_{\mathrm{SB}} = H_{\mathrm{L}} = 0$, so that $H$ has a $\mathbb{Z}_2^K$ symmetry which necessarily delocalizes $|\psi\rangle$ among all codewords.  A key observation is that even this model violates a key postulate of quantum statistical mechanics.  Adjacent eigenvalues $E_1 < E_2$ strongly cluster on the real line, in contrast to the predictions of random matrix theory that nearby eigenvalues will repel \cite{}. Here, every eigenvalue of $H$ belongs to a clustered pair obeying \begin{equation}
%     |E_1-E_2| \le \exp\left[-? \log \frac{?}{\epsilon} \cdot N\right] =  \exp[-\lambda N] \label{eq:eigenvalueattraction}
% \end{equation} 
% in the $\mathbb{Z}_2^K$-symmetric model.  Eigenvalues are drastically more clustered than the prediction of random matrix theory: $|E_1-E_2| \sim 2^{-N}$. 
% Upon returning to the full $H$, we prove that extreme eigenvalue clustering (\ref{eq:eigenvalueattraction}) \emph{must} occur if any eigenstate $|\psi\rangle$ is not fully localized near a single codeword.  In other words, if (\ref{eq:eigenvalueattraction}) does not hold, then for arbitrary $|\psi_0\rangle$, there must exist a codeword $\mathbf{z}_*$ for which \begin{equation}
%     \sum_{n=0}^{N_*} |c_{\mathbf{z}_*n}|^2 \ge 1 - \exp[-\mathrm{O}(N)].
% \end{equation}

It remains to show that the eigenstate is trapped near \emph{one} codeword.   First, we show that \emph{if} an eigenstate is not localized near a single codeword, there is an unlikely energetic resonance between different wells.  More precisely, we look at a truncated version of the Hamiltonian $H_{\mathbf{z}}$ which is isolated near codeword $\mathbf{z}$, and show that if $|\psi_0\rangle$ has comparable weight near codewords $\mathbf{z}$ and $\mathbf{z}^\prime$, then $H_{\mathbf{z}}$ and $H_{\mathbf{z}^\prime}$ must have respective eigenvalues $E$ and $E^\prime$ obeying $|E-E^\prime| \le 2^{-(2+c)N}$ for some $c>0$.  Intuitively, this is absurdly unlikely to happen for a generic Hamiltonian $H$; for example, in a chaotic system, random matrix theory predicts that nearby eigenvalues of a many-body Hamiltonian repel, implying energy splittings of at least $2^{-N}$ \cite{huse2007,bohigas}. To prove that it is impossible for one specific $H$, however, is challenging.  At this point, we invoke the disorder in $H_{\mathrm{L}}$ to show that such resonances between any two $E$ and $E^\prime$ are finely-tuned, and in particular that almost surely any disorder configuration we find has no such resonances.  Intuitively, the disorder easily splits resonances because certain linear combinations of $h_i$ are fields that tend to raise the energy of particular codewords, analogous to external magnetic fields in the Ising model: see Figure \ref{fig:wells}c.  Upon showing the absence of resonances, we have proved that every eigenstate is localized in an exponentially small fraction of the low-energy configuration space.

We have elected to work with c3LTC $H_0$ in the discussion above, which ensures all low energy configurations of $H_0$ are close to a codeword; this choice makes our calculation more pedagogical. However, localization should not be understood as a mere consequence of explicitly breaking a spontaneously broken $\mathbb{Z}_2^K$ symmetry associated to the original logical codewords of $H_0$.  Firstly, $H_{\mathrm{SB}}$ is not  perturbatively small, and can be chosen more generally than (\ref{eq:HSB}), so   $V+H_{\mathrm{L}}$ perturb a non-symmetric Hamiltonian. Secondly, we explain in the SM that MBL persists if $H_0$ is a more general LDPC code with $K,D=\mathrm{O}(N)$, but without LT, such that the vast majority of low energy states are far from all codewords.  Finally, the mathematical mechanism of localization is akin to Anderson's locator expansion \cite{anderson}, and arises from energy level detuning \cite{maybodorova}.  We achieve this without requiring $\exp[\mathrm{O}(N)]$ random couplings in $H$, in contrast to \cite{abouchacra}. %Lastly, if we do consider a model with spontaneous symmetry breaking, by choosing $H_{\mathrm{SB}}=H_{\mathrm{L}}=0$ and $[V,X_{\mathbf{z}}]=0$, we find that all low energy eigenvalues of $H$ are \emph{attracted} to nearby eigenvalues, violating the expectations from statistical mechanics and random matrix theory \cite{fratus,winer}. After adding explicit symmetry breaking, it remains non-trivial for all eigenstates to be localized: a natural expectation is that ``false vacua" \cite{coleman} trapped near one codeword could eventually tunnel into another codeword. Our MBL construction also implies ``false vacua" with infinite lifetime exist at finite $N$.

%  Previous efforts to show MBL in disordered spin chains aim to construct a complicated quasilocal unitary transformation that diagonalizes $H$ \cite{Imbrie_2016}; amusingly, this approach relies on an absence of level attraction, which also was a possible obstruction in our proof.  Crucially, the linear soundness of the c3LTC assured that we could avoid such level attraction.  These models then represent a rare instance where strong, non-perturbative statements can be proved about the spectrum of a many-body Hamiltonian without exactly diagonalizing it.

%The methods employed above are quite different from typical methods used to prove prethermalization (the slow onset of ergodicity) \cite{}; such methods were the backbone of the previous derivation of MBL in strongly disordered 1d systems \cite{}; however, this proof remained incomplete due to  explicit assumptions about many-body level statistics.  In our model, in contrast, the linear confinement of the c3LTC is sufficiently strong to prove the impossibility of any unusual resonances that could destroy MBL. 

\emph{Outlook}.---  We have found a $q$-local many-body quantum system for which all low energy-density eigenstates are localized. This existence proof  settles a major open problem in mathematical physics.   Intriguingly, \revise{while our model is certainly many-body, and exhibits localization, the models we have studied look nothing like  models usually explored in the MBL literature} -- rather than looking at highly disordered spin chains, we studied quantum-fluctuating classical error correcting codes, which cannot be embedded in $\mathrm{O}(1)$ spatial dimensions.  It is an important open problem to either construct a model which has provable MBL in a fixed spatial dimension, or to show the impossibility. Low dimensional problems with long-range interactions \cite{LRMBL, LRMBLnumerics} may be promising in this regard.   We comment that the existence of local integrals of motion \cite{liom,lbits} in our model is an open question, although it is unlikely to be so \cite{baldwin, geraedts}.  Our rigorous results are a starting point for future investigations.   %Looking forward, we expect that the model (\ref{eq:H}), and/or the methods described in this work, could find further applications in proving exact results about eigenvalue statistics, symmetry breaking, and chaos in many-body quantum systems.
%will  also then be important to understand the extent to which the infinite-dimensional MBL of this work is, or is not, related to the conjectured mechanisms for MBL in disordered spin chains.}

Looking forward, we expect our model and/or methods to provide a powerful new route to exact results in quantum statistical mechanics. For example, our methods lead to intriguing exact results about level statistics in a model with spontaneous symmetry breaking: upon restoring symmetry to our Hamiltonian (associated to bit flips corresponding to the codewords) by setting $H_{\mathrm{SB}}=H_{\mathrm{L}}=0$ and choosing any $V$ with $[V,X_{\mathbf{z}}]=0$, our methods directly show that eigenvalue splittings are anomalously small \cite{bozhao,Altshuler_2010}, in contrast to random matrix theory predictions for the spectral form factor \cite{fratus,winer}.   c3LTCs are closely related to good quantum LDPC codes \cite{Panteleev_2021,dinur2022good,dinur2022LTC,lin2022,Rakovszky:2023fng,Rakovszky:2024iks}, which are  quantum memories \cite{Hong:2024vlr}; the fate of such models under perturbations is worth investigating, as error correcting codes have already been shown to give rise to intriguing phases of classical \cite{Montanari_2006} and quantum \cite{Rakovszky:2023fng,Rakovszky:2024iks,Hong:2024vlr} matter.
 %such quantum code Hamiltonians may also give MBL.

\emph{Note added.}--- Other authors have been independently studying a similar model where similar phenomenology is found \cite{Breuckmann_2024}.

\emph{Acknowledgements.}--- We thank Chris Baldwin, Aaron Friedman, Yifan Hong, David Huse, Vedika Khemani and Chris Laumann for useful discussions.  This work was supported by the Alfred P. Sloan Foundation under Grant FG-2020-13795 (AL), the Department of Energy under Quantum Pathfinder Grant DE-SC0024324 (CY, AL), and the Air Force Office of Scientific Research under Award No. FA9550-20-1-0222 (RN).

\bibliography{ltcmbl}

%\comment{ % crop the SM
\onecolumngrid

\newpage

%\begin{appendix}
\setcounter{equation}{0}
\setcounter{page}{1}
\renewcommand{\theequation}{S\arabic{equation}}
\renewcommand{\thesubsection}{S\arabic{subsection}}
\renewcommand{\thepage}{S\arabic{page}}

\begin{center}
    {\large \textbf{Supplementary Material}}
\end{center}

%\section{Mathematical details} \label{app:proof}
%This appendix provides the details of the proof of MBL summarized in the main text.

\subsection{Formal statement of the main theorem}
We follow the notation of the main text.  Given that the perturbation $V$ is $\Delta'$-local, we define $\Delta:=\Delta'q$, and
choose the cutoff energy $N_*$ by \begin{equation}\label{eq:well<D}
         N_*:= \Delta n_*,
\end{equation}
where \begin{equation}
    n_*:=\left\lfloor \frac{\alpha(D-\Delta'-1)}{2\Delta} \right\rfloor,
\end{equation}
and $\lfloor\cdot \rfloor$ is the floor function.
Defining $D=d_0 N$ and $K=k_0 N$ for some O(1) $0<d_0,k_0<1$, we note that for sufficiently large $N$, \begin{equation}
    \frac{N_*}{N} = \mu > \frac{\alpha d_0}{3}.
\end{equation}
We then define projection operators onto one well of codeword $\bz$: \begin{equation}\label{eq:Pz=}
    P_{\bz} := \sum_{\bs:|\bs-\bz|\le (N_*-1)/\alpha} \ket{\bs} \bra{\bs}.
\end{equation}
Because the codewords have Hamming distance $D$,
these projectors are orthogonal to each other: $P_{\bz}P_{\bz'}=0 \quad (\forall \bz\neq \bz')$ because two bitstrings $\bs,\bs'$ in two wells satisfy \begin{equation}\label{eq:s-s'>D}
    \abs{\bs-\bs'}\ge \abs{\bz-\bz'} - |\bs-\bz|-|\bs'-\bz'|\ge D-2\frac{D-\Delta'-1}{2} \ge \Delta'+1>0.
\end{equation}
I.e. the Hilbert subspace of each well is orthogonal to those of other wells. In fact, \eqref{eq:s-s'>D} implies that $V$ does not couple the wells directly: for all $\bz\neq \bz'$, \begin{equation}\label{eq:V_connect_z}
    P_{\bz}VP_{\bz'}=0 .
\end{equation}
With these facts collected, we can state the main result formally:

\begin{thm}\label{thm:main}
     Let $H$ be defined by (\ref{eq:H}).  Suppose $V$ is $\Delta'$-local with \begin{equation}\label{eq:V<eps}
         \norm{V}\le \epsilon N,
     \end{equation}
     where the perturbation strength is bounded by \begin{equation}\label{eq:eps<N*}
        \epsilon \le \frac{\mu}{300}\times \min\lr{1, 7\times 2^{-\frac{5\Delta}{2\mu}\lr{1+\frac{k_0}{4}+\delta} } },
    \end{equation}
    where $\delta>0$ is any positive constant independent of $N$.
    Suppose $N$ is sufficiently large such that \begin{equation}\label{eq:N>constant}
        N_* \ge 450 \Delta.
    \end{equation}
    Suppose the longitudinal fields $\{h_i\}$ are independent and identically distributed zero-mean, unit variance Gaussian random variables.
    % \begin{equation}
    %     f(h_i) = \frac{1}{\sqrt{2\pi}}\ee^{-h_i^2/2}.
    % \end{equation}
    With high probability \begin{equation}\label{eq:probability0}
        \mathrm{Pr}[H \text{ has low energy localization}]\ge 1- 2N\epsilon^{-2} 2^{-\delta N}-\ee^{-N^2/4},
    \end{equation}
both the random $H_{\mathrm{L}}$ has bounded norm $\lVert H_{\mathrm{L}}\rVert \le \epsilon N$, and    all eigenstates $|\psi\rangle$ of $H$ with energy eigenvalue bounded by \begin{equation}\label{eq:E<N*}
       E< E_* := \frac{N_*}{30} = \frac{\mu}{30}N,
    \end{equation}
    are trapped near some codeword $\bz(\psi)$: \begin{equation}\label{eq:1-Pz<exp}
        \norm{ \lr{1-P_{\bz(\psi)}} \kpsi}\le \sqrt{2}N \ee^{-\delta N}.
    \end{equation} 
\end{thm}

Here the $\mathrm{O}(1)$ constants are chosen for convenience and are not intended to be optimal. Note that there are exponentially many eigenstates in the energy window \eqref{eq:E<N*}: All of the $2^K$ codewords have energy zero for the $A=H_0+H_{\rm{SB}}$ part of the Hamiltonian, and the other part $B=V+H_{\rm L}$ has operator norm bounded by $\norm{B}\le 2\epsilon N$. According to Weyl's inequality, \begin{equation}
    \abs{\lambda_m(A+B)-\lambda_m(A)} \le \norm{B},
\end{equation}
where $\lambda_m(A)$ is the $m$-th smallest eigenvalue of $A$.  There are at least $2^K$ eigenstates of $H$ with energy $E\le 2\epsilon N\le E_*/5 <E_*$. Extending this argument to bitstrings near codewords but with finite energy density $<E_*/N-2\epsilon$ with respect to $A$ leads to many more eigenstates in the energy window \eqref{eq:E<N*}. The total number of such eigenstates is $\gtrsim 2^K \times {N \choose (E_*-2\epsilon N)/q}$, because starting from each codeword, one can flip any subset $S$ of bits with $|S|\le (E_*-2\epsilon N)/q$ and still be at sufficiently low energy. % \andy{we can get a better count than $2^K$, something like $N$ choose $E_* - 2\epsilon N$ multiplied by $2^K$...} \chao{OK added.}

We will present the proof of Theorem \ref{thm:main} in Appendix \ref{app:proof}, which invokes a crucial Lemma \ref{lem:no_degen} about the energy spectrum that we will prove in the last Appendix \ref{app:no_degen}. Prior to the main proof, it is useful to first show a weaker result in Appendix \ref{app:all_codeword}, which proves all low-energy eigenstates are trapped near codewords, although the number of codewords may be $>1$.  Appendix \ref{app:time} proves that many-body states, even when not initialized in eigenstates, are trapped near codewords for infinite time, albeit with a more stringent bound on $\epsilon$.   Lastly, in Appendix \ref{app:extend} we discuss (briefly) how to extend our result to more general LDPC codes beyond c3LTCs, and point out that the general conclusion does not change.

\subsection{Low-energy eigenstates are trapped near codewords}\label{app:all_codeword}

We first introduce some notation and overview some simple facts.
For any operator $V'$ that is $\Delta'$-local, acting it on any bitstring $\ket{\bs}$ changes its $H_0$-energy by at most $\Delta=\Delta'q$: \begin{equation}\label{eq:V_connect_Delta}
    \abs{\alr{H_0}_{\bs'}-\alr{H_0}_{\bs}} \le \Delta, \quad \text{ for all } \bs' \text{ such that  } \alr{\bs'|V'|\bs}\neq 0,
\end{equation} 
where $\alr{A}_{\psi}:=\bra{\psi}A\kpsi$. 
Any state $\kpsi$ can be decomposed into eigenstates of $H_0$, which we group into bins separated by $\Delta$ for later convenience: \begin{equation}\label{eq:psi=psin} 
    \kpsi = \sum_{n=1}^{n_*+1} c_n \kpsin{n},
\end{equation} 
where $0\le c_n\le 1$, and $\kpsin{n}$ with $n\le n_*$ ($n=n_*+1$) is a normalized state in the eigen-subspace of $H_0$ with energy $(n-1)\Delta,(n-1)\Delta+1,\cdots,n\Delta-1$ ($\Delta n_*,\Delta n_*+1,\cdots$). Due to the locally testable condition, the $n\le n_*$ part is exactly the direct sum of the well subspaces $P_\bz$ defined by \eqref{eq:Pz=}. Any operator $V'$ satisfying \eqref{eq:V_connect_Delta} can be expressed as a block-tridiagonal matrix \begin{equation}\label{eq:V=Vnn}
    V' = \sum_{n,n':|n-n'|\le 1} V'_{nn'},
\end{equation}
which only connects subspaces labeled by $n$ for neighboring $n$s.

\begin{prop}\label{prop:cn<}
    For any Hermitian operator $V'$ that satisfies \eqref{eq:V_connect_Delta} with \begin{equation}\label{eq:V'<eps}
        \norm{V'}\le 2\epsilon N,
    \end{equation}
    where $\epsilon$ is bounded by \eqref{eq:eps<N*} and $N$ is bounded by \eqref{eq:N>constant}, any eigenstate $\psi$ of $H_0 + H_{\rm SB} + V'$ with eigenvalue $E$ satisfying \eqref{eq:E<N*} is trapped near codewords:   \begin{equation}
        c_{n_*+1}\le  2^{-\lambda N},
    \end{equation}
    where $c_n$ is the amplitude in decomposition \eqref{eq:psi=psin}, and \begin{equation}\label{eq:lambda=eps}
        \lambda = \frac{4\mu}{5\Delta } \log_2 \frac{7\mu}{300\epsilon}.
    \end{equation}
\end{prop}

\begin{proof}
Denote $H_0':=H_0 + H_{\rm SB}$, which satisfies \begin{equation}\label{eq:H0'>H0}
    \alr{H_0'}_\phi \ge \frac{1}{2}\alr{H_0}_\phi
\end{equation}
for all $\ket{\phi}$, 
because terms in $H_0'$ corresponding to a given parity check are bounded so by the corresponding parity check term in $H_0$, due to %\chao{deleted ``linear soundness (\ref{eq:confinement})''. } 
constraint (\ref{eq:HSB}) on $H_{\mathrm{SB}}$.
Plugging \eqref{eq:psi=psin} into the eigenvalue equation $(H_0' + V')\kpsi = E\kpsi$, we have \begin{subequations}
\begin{align}
    (H_0'+V'_{n_*+1,n_*+1}-E) c_{n_*+1} \kpsin{n_*+1} &= - V'_{n_*+1,n_*} c_{n_*}\kpsin{n_*}, \label{eq:bN+1} \\
    (H_0'+V'_{nn}-E) c_n \kpsin{n} + V'_{n,n+1}c_{n+1}\kpsin{n+1} &= - V'_{n,n-1} c_{n-1}\kpsin{n-1}, \quad 2\le n\le n_*. \label{eq:n_eigeneq}
\end{align}
\end{subequations}

Taking inner product with $\kpsin{n_*+1}$, \eqref{eq:bN+1} yields \begin{align}
\mlr{\frac{\Delta n_*}{2}-2\epsilon N - E}c_{n_*+1} &\le (\alr{H_0'}_{\psi_{n_*+1}}+\alr{V'}_{\psi_{n_*+1}} - E)c_{n_*+1} = - \alr{\psi_{n_*+1}|V'|\psi_{n_*}} c_{n_*} \le 2\epsilon N c_{n_*}, \nonumber\\
    c_{n_*+1} &\le \frac{4\epsilon N}{\Delta n_*-2E-4\epsilon N} c_{n_*}.\label{eq:bN+1<bN}
\end{align}
Here in the first line, we have used \eqref{eq:H0'>H0} and \eqref{eq:V'<eps} with $\norm{V'_{nn'}}\le \norm{V'}$; to get the second line we have used \eqref{eq:eps<N*} and \eqref{eq:E<N*} so that the denominator in the last expression is positive. Note that the first line of \eqref{eq:bN+1<bN} also implies $\alr{\psi_{n_*+1}|V'|\psi_{n_*}}$ is real, as it is the only possibly complex-valued coefficient in the equation $\langle \psi_{n_*+1}|H-E|\psi\rangle = 0$. The conjugate of this matrix element (which is evidently itself) appears on the left hand side of \eqref{eq:n_eigeneq} for $n=n_*$ when taking inner product with $\kpsin{n_*}$, which further implies the right-hand-side matrix element $\alr{\psi_{n_*}|V'|\psi_{n_*-1}}$ is also real. This procedure can be iterated to conclude that $\alr{\psi_{n}|V'|\psi_{n-1}}$ is real for all $n\ge 1$; in other words, even if in the original computational $V$ is not real-valued, the $|\psi_n\rangle$s themselves necessarily absorb all complex phases.

We obtain bounds like \eqref{eq:bN+1<bN} in a similar way: Taking inner product with $\kpsin{n}$, \eqref{eq:n_eigeneq} becomes \begin{equation}\label{eq:n_ineq}
    \mlr{\frac{\Delta(n-1)}{2} - E- 2\epsilon N}c_n - 2\epsilon N c_{n+1} \le 2\epsilon N c_{n-1},
\end{equation}
using $\alr{\psi_{n}|V'|\psi_{n+1}}\ge -\norm{V'}\ge -2\epsilon N$.
For example, plugging in \eqref{eq:bN+1<bN} with $n=n_*$ yields \begin{align}
    4\epsilon N c_{n_*-1} &\ge \mlr{\Delta(n_*-1) - 2E- 4\epsilon N \lr{1+\frac{4\epsilon N}{\Delta n_*-2E-4\epsilon N}} }c_{n_*} \ge \mlr{\Delta(n_*-1) - 2E- 4\epsilon N\times 2 }c_{n_*}, \nonumber\\ c_{n_*}&\le \frac{4\epsilon N}{\Delta(n_*-1) - 2E- 8\epsilon N} c_{n_*-1},
\end{align}
which has nearly the same form as \eqref{eq:bN+1<bN}. 

We can iterate the above process: Suppose \begin{equation}\label{eq:cn+1<cn}
    c_{n+1}\le c_n,
\end{equation}
which holds at the ``initial" (largest) $n=n_*$, then \eqref{eq:n_ineq} leads to \begin{equation}\label{eq:cn<cn-1}
    c_n \le \frac{4\epsilon N}{\Delta(n-1) - 2E- 8\epsilon N} c_{n-1},
\end{equation}
which justifies condition \eqref{eq:cn+1<cn} for the next step, as long as \begin{equation}
    n \ge n_{\rm stop}:= 2+ \left\lfloor \frac{2(E+6\epsilon N)}{\Delta} \right\rfloor,
\end{equation}
so that the denominator in $\eqref{eq:cn<cn-1}$ is no smaller than the numerator.
We stop the iteration at $n=n_{\rm stop}$, so that \eqref{eq:cn<cn-1} for all $n\ge n_{\rm stop}$ yield \begin{equation}\label{eq:cn*<prodn}
    c_{n_*+1}\le c_{n_{\rm stop}-1}  \prod_{n=n_{\rm stop}}^{n_*+1} \frac{4\epsilon N}{\Delta(n-1) - 2E- 8\epsilon N} \le \prod_{n=n_{\rm stop}}^{n_*+1} \frac{4\epsilon N}{\Delta(n-1) - 2E- 8\epsilon N}.
\end{equation}

We simplify \eqref{eq:cn*<prodn} by keeping only factors from $n\ge n_{\rm stop}+n'$ with $n'=\lfloor (n_*-n_{\rm stop})/10\rfloor+2$: \begin{align}
    c_{n_*} &\le \lr{ \frac{4\epsilon N}{\Delta(n_{\rm stop}+n'-1) - 2E- 8\epsilon N}}^{n_*-n_{\rm stop} - n'+2} \le \lr{ \frac{4\epsilon N}{\Delta(9n_{\rm stop}+n_*)/10 - 2E- 8\epsilon N }}^{9\lr{n_*-n_{\rm stop}}/10} \nonumber\\
    &\le \lr{ \frac{4\epsilon N}{\mlr{N_*+18(E+6\epsilon N)+9\Delta}/10 - 2E- 8\epsilon N }}^{9\mlr{N_*/\Delta-2-2(E+6\epsilon N)/\Delta}/10}\nonumber\\
    &\le \lr{\frac{40\epsilon }{\mu(1 - 2E/N_*) }} ^{9(N_*-2E-12\epsilon N-2\Delta)/(10\Delta )} \nonumber\\
    &\le \lr{2^{-5\Delta \lambda /(4\mu) }}^{ 9\lr{N_*-\frac{1}{15}N_* - \frac{12}{300}N_* - 2\Delta}/(10\Delta)} \le \lr{2^{-5\Delta \lambda /(4\mu) }}^{ 8N_*/(10\Delta)} \le 2^{-\lambda N}.
\end{align}
In the last line, we used \eqref{eq:eps<N*}, \eqref{eq:E<N*} and \eqref{eq:lambda=eps}, so that the denominator is bounded by $\mu(1 - 2E/N_*) \ge 14\mu/15$. 
\end{proof}

\subsection{Proof of the main theorem}\label{app:proof}
\begin{proof}
Let $\kpsi$ be any eigenstate of $H$ with low energy \eqref{eq:E<N*}. The decomposition \eqref{eq:psi=psin} can be organized as \begin{equation}
    \kpsi = \ket{\psi_{n_*+1}}+ \sum_{\text{codeword } \bz} \ket{\psi_\bz} ,
\end{equation}
where $\ket{\psi_\bz}$ is in the well of codeword $\bz$: $P_{\bz'} \ket{\psi_\bz}=\delta_{\bz'\bz}\ket{\psi_\bz} $, and $\ket{\psi_{n_*+1}}$ is the part outside of any well. Similarly, $H$ can be expanded as \begin{equation}
    H = H_>+ \sum_{\text{codeword } \bz} H_\bz,
\end{equation}
where $H_\bz = P_{\bz} H P_{\bz}$ is the Hamiltonian restricted in well $\bz$, and $H_>$ is everything else so that $P_{\bz} H_> P_{\bz}=0$. Acting $P_\bz$ on the eigenvalue equation $H\kpsi = E\kpsi$, we have \begin{equation}\label{eq:Hpsiz=}
    (H_\bz-E) \ket{\psi_\bz}= -P_\bz H_> \kpsi = -P_\bz H_> (1-P_\bz) \kpsi = -P_\bz V (1-P_\bz) \kpsi= -P_\bz V \ket{\psi_{n_*+1}},
\end{equation}
because only $V$ in $H$ can map a state out of a well, and it only maps it out to the $n=n_*+1$ subspace (see \eqref{eq:V_connect_z}).

According to \eqref{eq:Hpsiz=}, for any $\bz$ with nonzero support $\ket{\psi_\bz}\neq 0$, we have either $E$ is an eigenvalue of $H_\bz$, or $H_\bz-E$ is invertible so that \begin{equation}
    \ket{\psi_\bz} = -(H_\bz-E)^{-1}P_\bz V \ket{\psi_{n_*+1}}, 
\end{equation}
which implies that \begin{equation}
    \norm{(H_\bz-E)^{-1}}\ge \frac{\norm{\ket{\psi_\bz}}}{ \norm{V \ket{\psi_{n_*+1}}}} \ge \frac{\norm{\ket{\psi_\bz}}}{\epsilon N c_{n_*+1}},
\end{equation}
where we have used $\norm{A \ket{\phi}} \le \norm{A} \norm{\ket{\phi}}$, \eqref{eq:V<eps}, and $\norm{P_\bz}=1$ because it is a projector. In both cases, $E$ is close to an eigenvalue $E_{\bz, m}$ ($m=1,2,\cdots$ is an index) of $H_\bz$: there exists $m$ such that  \begin{equation}\label{eq:E-Ezm<psi}
     \abs{E-E_{\bz, m}} \le \frac{\epsilon N 2^{-\lambda N}}{\norm{\ket{\psi_\bz}}}.
\end{equation}
Here we have applied Proposition \ref{prop:cn<} to $V'=V+H_{\rm L}$ that satisfies \eqref{eq:V'<eps} with high probability \begin{equation}\label{eq:prob_HL}
    \mathrm{Pr}[\norm{H_{\rm L}}\le \epsilon N] \ge 1- \ee^{-N^2/4},
\end{equation}
because \begin{equation}
    \norm{H_{\rm L}}^2 = \frac{\epsilon^2}{N} \lr{\sum_i \abs{h_i}}^2 \le \epsilon^2 h^2,
\end{equation}
where random variable $h:=\sqrt{\sum_i h_i^2}$ has the following probability density function $P(h)$ on $h\in \mathbb{R}^+$ known as the chi distribution: \begin{align}
    P(h) &= \frac{2^{1-N/2}}{\Gamma(N/2)} \ee^{-h^2/2} h^{N-1}. 
\end{align}
We also notice that if $h>N$ and $N\ge 450$ from \eqref{eq:N>constant}, %\rahul{needs intermediate steps. I am guessing you have used asymptotic form of the Gamma function and Stirling, and set $h=N$, but that alone does not seem sufficient? How do you turn $e^{-h^2/2}$ into $e^{-N^2/4}$?} \chao{added the step below?}
\begin{equation}
   P(h) \le \ee^{-N^2/4} \ee^{-h}, \quad \text{if } h> N,
\end{equation}
because $P(h)\ee^{N^2/4} \ee^{h} \le \exp[-h^2/2+(N-1)\log h + \lr{h^2/4+h}]=\exp[-h^2/4+(N-1)\log h +h]$ decays to zero very quickly at large $h$.
Integrating $P(h)$ in the range $h\le N$ yields \eqref{eq:prob_HL}.

As a result, \eqref{eq:E-Ezm<psi} implies that for any $\bz$ with \begin{equation}\label{eq:psiz>2N}
    \norm{\ket{\psi_\bz}} \ge N 2^{-(\frac{1}{2}k_0+\delta) N},
\end{equation}
$E$ is extremely close to an eigenvalue of $H_\bz$: for some $m$,
\begin{equation}\label{eq:E-Ezm<}
    \abs{E-E_{\bz, m}} \le 2^{-(\lambda-\frac{1}{2}k_0-\delta)N}\epsilon^{-1}.
\end{equation}

We then invoke the following Lemma, which is proved in the final subsection:
\begin{lem}\label{lem:no_degen}
    Consider any constant $\lambda'>2$. With probability \begin{equation}\label{eq:probability}
        p\ge 1-2N\epsilon^{-2}2^{-(\lambda'-2)N},
    \end{equation}
there are no degeneracies among distinct $H_\bz$: if  $\bz\neq \bz'$, for all $m,m'$, \begin{equation}\label{eq:no_degen}
          \abs{E_{\bz,m} - E_{\bz',m'}} \ge 3\epsilon^{-1} 2^{-\lambda' N}.
    \end{equation}
\end{lem}

Plugging the second bound of $\epsilon$ in \eqref{eq:lambda=eps} yields $\lambda \ge 2+\frac{1}{2}k_0+2\delta$, so that \eqref{eq:probability} with $\lambda'=\lambda-\frac{1}{2}k_0-\delta$ becomes the first two terms in \eqref{eq:probability0}. As a result, with high probability \eqref{eq:probability0}, both $\norm{H_{\rm L}}\le \epsilon N$ from \eqref{eq:prob_HL} and \eqref{eq:no_degen} hold. For each eigenvalue $E$ of the full $H$ at low energy, there can be at most one codeword $\bz$ such that both \eqref{eq:E-Ezm<} and \eqref{eq:no_degen} holds; i.e. there can be at most one codeword $\bz$ such that \eqref{eq:psiz>2N} holds. Because $\kpsi$ is normalized, there is then exactly one such $\bz=\bz(\psi)$ where $\kpsi$ is trapped; the leakage out of this well is \begin{equation}
    \norm{ \lr{1-P_{\bz(\psi)}} \kpsi}^2 \le c_{n_*+1}^2 + (2^K-1)\lr{N 2^{-(\frac{1}{2}k_0+\delta) N}}^2 \le 2^{-2\lambda N}+N^2 2^{-2\delta N} \le 2N^2\ee^{-2\delta N}, \label{eq:leak}
\end{equation}
which leads to \eqref{eq:1-Pz<exp}.
\end{proof}

\subsection{Proof of Lemma \ref{lem:no_degen}: no degeneracy among wells with high probability}\label{app:no_degen}

\begin{proof}[Proof of Lemma \ref{lem:no_degen}]
In this proof, we denote $E_{\bz,m}=E_{\bz,m}(h)$ where $h:=\{h_i\}$ denotes the random variables from $H_{\mathrm{L}}$. The corresponding eigenvectors are $\ket{\phi_{\bz,m}(h)}$.

First, let us focus on one well $\bz$ to show that its energies $E_{\bz,m}(h)$ can be shifted by tuning $h$ to avoid degeneracy with other wells $\bz'\neq \bz$. The random variables $h_i$ can be denoted as a $N$-dimensional vector $|h)$ with entries $(i|h) = h_i$. Instead of the original $\{|i):i=1,\cdots,N\}$, we use an alternative orthonormal basis $\{|\bz;k): k=0,\cdots,N-1\}$ determined by $\bz$, that is some Fourier transform of the original one: \begin{equation}\label{eq:Fourier}
    |\bz;k):=\frac{1}{\sqrt{N}} \sum_i (-1)^{z_i} \ee^{\ii \frac{2\pi}{N}ki} |i), \quad \Leftrightarrow\quad  |i) = (-1)^{z_i} \frac{1}{\sqrt{N}} \sum_i  \ee^{-\ii \frac{2\pi}{N}ki} |\bz;k).
\end{equation}
This basis corresponds to variables $\{h_\bz, \overline{h}_{\bz}\}$, where $h_{\bz}:= (\bz;0|h)=\sum_i (-1)^{\bz_i} h_i$ and $\overline{h}_{\bz}$ denotes variables from inner product with the rest $|\bz;k)$ with $k>0$. More precisely, for $k\neq N/2$ (if $N$ is even), the inner products with $|\bz;k)$ and $|\bz;N-k)$ yield a complex-conjugate pair, whose real and imaginary parts are two real variables in $\overline{h}_{\bz}$. As we have simply made an orthogonal transformation from $\{h_i\}$ to $\{h_\bz, \overline{h}_{\bz}\}$, clearly the new variables are independent and identically distributed Gaussian random variables with the same probability distribution as before.  Moreover, $h_{\bz}$ is independent from $\overline{h}_{\bz}$. 

% Since the Jacobian of this variable change is one, the new variables still obey a Gaussian probability distribution \begin{equation}\label{eq:Phz}
%     P(h_\bz, \overline{h}_{\bz}) = P(\{h_i\})=\lr{2\pi}^{-N/2} \ee^{-\frac{1}{2}\sum_i h_i^2} = \lr{2\pi}^{-N/2} \ee^{-\frac{1}{2}\lr{h_\bz^2 + \overline{h}_{\bz}^2 }} =: P(h_\bz)P( \overline{h}_{\bz}),
% \end{equation}
% where $\overline{h}_{\bz}^2$ denotes the sum of square of the other variables. I.e. $h_\bz$ and $\overline{h}_{\bz}$ are independent, with $P(h_\bz)$ being the standard normal distribution.

Changing $h_\bz$ while keeping $\overline{h}_{\bz}$ fixed, the partial derivative is \begin{equation}
    \partial_{h_{\bz}} = \sum_i \frac{\partial h_i}{\partial {h_{\bz}}} \partial_{h_i} = \frac{1}{\sqrt{N}} \sum_i (-1)^{\bz_i}\partial_{h_i},
\end{equation} 
from the second transformation in \eqref{eq:Fourier}.
The Feynman-Hellmann theorem then yields \begin{equation}\label{eq:feynman}
    \partial_{h_{\bz}} E_{\bz',m'}(h) = \alr{\partial_{h_{\bz}} (P_{\bz'}HP_{\bz'} )}_{\phi_{\bz',m'}(h)} = \alr{\partial_{h_{\bz}} H}_{\phi_{\bz',m'}(h)} = \frac{\epsilon}{N}\alr{Z_\bz}_{\phi_{\bz',m'}(h)},
\end{equation}
where \begin{equation}
    Z_\bz := \sum_i (-1)^{\bz_i} Z_i,
\end{equation}
and we have used $\partial_{h_i} P_{\bz'}=0$ and $P_{\bz'}\ket{\phi_{\bz',m'}(h)} = \ket{\phi_{\bz',m'}(h)}$. Note that \eqref{eq:feynman} holds also for $\bz'=\bz$. $Z_\bz$ is close to its maximum $N$ for any state in the well $\bz$, while it is far from maximum for any state in the other wells due to the macroscopic distance $D$. For any pair $E_{\bz,m}(h), E_{\bz',m'}(h)$ where $\bz'\neq \bz$, this leads to
\begin{align}
    \partial_{h_{\bz}}\mlr{ E_{\bz,m}(h)- E_{\bz',m'}(h)} &= \frac{\epsilon}{N} \lr{ \alr{Z_\bz}_{\phi_{\bz,m}(h)} - \alr{Z_\bz}_{\phi_{\bz',m'}(h)}} \ge \frac{\epsilon}{N} \mlr{\left(N-\frac{N_*}{\alpha}\right) - \left(N-D+\frac{N_*}{\alpha}\right)} \ge \frac{\epsilon}{N},
\end{align}
similar to \eqref{eq:s-s'>D}. Therefore, focusing on the single variable $h_{\bz}$ with the other $\overline{h}_\bz$ parameters fixed, the slope of the energy difference $ E_{\bz,m}(h)- E_{\bz',m'}(h)$ is strictly positive and larger than $\epsilon/N$. This energy difference can then cross $0$ at most once, and the interval of $h_\bz$ that violates \eqref{eq:no_degen} (for this particular pair of energies) has length \begin{equation}
    \Delta h_{\bz,m;\bz',m'}(\overline{h}_\bz)\le \Delta h := 3N\epsilon^{-2}2^{-\lambda' N}.
\end{equation}
%determined by $\overline{h}_\bz$ and which are the two energy levels. 
Since there are fewer than $2^{N-K}\times 2^N=2^{2N-K}$ energy pairs between an energy of well $\bz$ and any other well, there are at most $2^{2N-K}$ intervals of $h_\bz$ with length $\le \Delta h$ that violate \eqref{eq:no_degen}. Since these intervals may overlap, the total length of their union, i.e. the resonance region for $h_\bz$, is bounded by $2^{2N-K}\Delta h$. Therefore, given $\overline{h}_\bz$, the probability to find $h_\bz$ that lands in the resonance region is bounded by \begin{equation}\label{eq:phz<}
    p_{\rm each}=\int_{\text{resonance region}} \dd h_\bz P(h_\bz)  \le 2^{2N-K}\Delta h\cdot \max_h P(h_\bz) = \frac{2^{2N-K}\Delta h}{\sqrt{2\pi}},
\end{equation}
independent of $\overline{h}_\bz$.

The above analysis has fixed a particular well $\bz$, and shows a small probability \eqref{eq:phz<} to violate \eqref{eq:no_degen} that involves $\bz$. By the pigeonhole principle, any parameter $h$ that violates \eqref{eq:no_degen} should at least violate it in one well, so the total probability to violate \eqref{eq:no_degen} is \begin{equation}
    p_{\rm total}\le 2^K p_{\rm each}=\frac{3}{\sqrt{2\pi} \epsilon^2}N 2^{-(\lambda'-2)N}\le \frac{2N 2^{-(\lambda'-2)N}}{\epsilon^2},
\end{equation}
leading to probability \eqref{eq:probability} that there are no resonances between any wells.
\end{proof}

\subsection{Trapping states near codewords for all time}\label{app:time}

\begin{prop}
    Let $H$ be any Hamiltonian considered in Theorem \ref{thm:main} that has localized eigenstates below energy $E_*$ from \eqref{eq:E<N*} (which holds almost surely with high probability \eqref{eq:probability}) with parameter $\delta>1$. Any normalized initial state $\kpsi$ supported on bitstrings close to a given codeword $\bz$: \begin{equation}\label{eq:initial_state}
        \kpsi = \sum_{\bs:|\bs-\bz|\le E_*/(3q)} a_\bs \ket{\bs}, 
    \end{equation}
    remains trapped in the well $\bz$ forever: for any $t\in \mathbb{R}$, \begin{equation}\label{eq:evolve<exp}
        \norm{ \lr{1-P_\bz}\ee^{-\ii tH} \kpsi}\le \left(\frac{9}{10}\right)^{\mu N/225\Delta}+ 2\sqrt{2}N 2^{-(\delta-1) N}.
    \end{equation} 
\end{prop}

\begin{proof}
According to Theorem \ref{thm:main}, the eigenstates of $H$ with energy $E<E_*$ can be labeled by $\{\ket{\bz',m}\}$, where $\ket{\bz',m}$ is the $m$-th eigenstate trapped in well $\bz'$: \begin{equation}\label{eq:zm_trap}
    \norm{(1-P_{\bz'})\ket{\bz',m}} \le \sqrt{2}N \ee^{-\delta N}.
\end{equation}
Let $\widetilde{P}_>:=1-\sum_{\bz',m} \ket{\bz',m}\bra{\bz',m}$ be the projector onto $E\ge E_*$ eigenstates.

Expanding \begin{equation}
    \ee^{-\ii tH}\kpsi = \ee^{-\ii tH}\widetilde{P}_>\kpsi+ \sum_{\bz',m}a_{\bz',m}\ee^{-\ii tE_{\bz',m}}\ket{\bz',m}
\end{equation} with $a_{\bz',m}=\alr{\bz',m|\psi}$ and using the triangle inequality, we have \begin{align}\label{eq:evolve<highE}
    \norm{ \lr{1-P_\bz}\ee^{-\ii tH} \kpsi} &\le \norm{ \widetilde{P}_>\kpsi}+ \sum_m\norm{ \lr{1-P_\bz}a_{\bz,m} \ket{\bz,m}}+\sum_{\bz'\neq \bz,m}\abs{a_{\bz',m} } \nonumber\\
    &\le \norm{ \widetilde{P}_>\kpsi}+ 2^{N-K}\max_m \norm{ \lr{1-P_\bz} \ket{\bz,m}} + 2^N \max_{\bz'\neq \bz,m} \abs{\alr{\psi|\bz',m}} \nonumber\\
    &\le \norm{ \widetilde{P}_>\kpsi}+ \sqrt{2}N 2^{-\delta N}\lr{2^{N-K} + 2^N} \le \norm{ \widetilde{P}_>\kpsi}+2\sqrt{2}N 2^{-(\delta-1) N}.
\end{align}
where we have used $\norm{1-P_\bz},\norm{\ee^{-\ii tH}}\le 1$ in the first line, and \eqref{eq:zm_trap} together with $\abs{\alr{\psi|\bz',m}} \le \norm{(1-P_{\bz'})\ket{\bz',m}}$ for $\bz'\neq \bz$ in the last line, because $\kpsi$ is not contained in well $\bz'$.

It remains to bound the high-energy contribution, the first term in \eqref{eq:evolve<highE}. Observe that \begin{equation}\label{eq:Hpsi<E}
    \norm{H\kpsi}\le \norm{H_0'\kpsi}+\norm{V'\kpsi}\le \frac{3}{2} \norm{H_0\kpsi} + 2\epsilon N \le \frac{1}{2}E_* + 2\epsilon N,
\end{equation}
where we have used \eqref{eq:V'<eps} for $V'=V+H_{\rm L}$, which is satisfied for the chosen $H$ (see \eqref{eq:prob_HL}). We have also used $H_0'\ge 3H_0/2$ similar to \eqref{eq:H0'>H0}, and that $\kpsi$ given by \eqref{eq:initial_state} is supported in the subspace of energy $E_0\in [0,E_*/3]$ measured by $H_0$, because flipping one qubit violates at most $q$ more checks. Since $H\kpsi$ is supported in the subspace of energy $E_0\le E_*/3 + \Delta$ measured by $H_0$ because $V'$ is $\Delta'$-local, so $\norm{H^2\kpsi}\le \mlr{\frac{1}{2}(E_*+3\Delta) + 2\epsilon N}\norm{H\kpsi}\le \mlr{\frac{1}{2}(E_*+3\Delta) + 2\epsilon N}\lr{\frac{1}{2}E_* + 2\epsilon N}$ similarly as \eqref{eq:Hpsi<E}. Iterating this yields \begin{equation}\label{eq:Hkpsi<}
    \norm{H^k\kpsi} \le \prod_{k'=0}^{k-1} \mlr{\frac{1}{2}\lr{E_*+3\Delta k'} + 2\epsilon N} \le \lr{\frac{9}{10}E_*}^k,
\end{equation}
with \begin{equation}
    k= \left\lfloor\frac{4E_*}{15\Delta} -\frac{4\epsilon N}{3\Delta}\right\rfloor+1\ge \frac{\mu N}{3\Delta}\lr{\frac{4}{150}-\frac{4}{300}} = \frac{\mu N}{225\Delta},
\end{equation}
where we have used \eqref{eq:eps<N*}.
On the other hand, $\norm{H^k\kpsi}\ge \norm{\widetilde{P}_>H^k \kpsi}=\norm{H^k \widetilde{P}_>\kpsi}\ge E_*^k \norm{\widetilde{P}_>\kpsi}$, so \eqref{eq:Hkpsi<} yields \begin{equation}
    \norm{\widetilde{P}_>\kpsi} \le \left(\frac{9}{10}\right)^k \le \left(\frac{9}{10}\right)^{\mu N/(225\Delta)}.
\end{equation}
Combining this with \eqref{eq:evolve<highE} leads to \eqref{eq:evolve<exp}.
\end{proof}

We remark in passing that our methods require a tighter bound on $\epsilon$ to prove this freezing of \emph{arbitrary} quantum states near the bottom of a well.  Whether there is a physical regime of eigenstate localization, yet delocalization of a typical initially localized state, could be an interesting question to explore in future work.

\subsection{Beyond c3LTCs}\label{app:extend}

 We conclude by explaining why the c3LTC is a technically convenient, though not essential, ingredient in our analysis. We now consider a more general parity check matrix $\mathsf{H}$ associated with an LDPC code, which may even be non-redundant.  Suppose that $\mathsf{H}$ has \emph{linear confinement}, which implies that for all $|\mathbf{x}|\le \gamma N$ with a $\mathrm{O}(1)$ constant $\gamma$, \begin{equation}
     |\mathsf{H}\mathbf{x}| \ge \alpha |\mathbf{x}|. \label{eq:confine}
 \end{equation}
It is useful to shift $\gamma$ by $\mathrm{O}(N^{-1})$ to ensure $\gamma N$ is an integer. By linearity, (\ref{eq:confine}) holds for $\mathbf{x}$ within $\gamma N$ Hamming distance of any codeword.  Moreover, the same property holds for any possible ``sick configuration" of sufficiently low energy density: given linear confinement (\ref{eq:confine}) near a codeword, for any $\mathbf{x}$ which obeys \begin{equation}\label{eq:Hx<gamma}
        |\mathsf{H}\mathbf{x}| \le \frac{\alpha \gamma}{4} N-1,
    \end{equation}
    then for all $\mathbf{y}$ obeying $ |\mathbf{y} | \le \gamma N$, \begin{equation}\label{eq:Hx+y>Hy}
        |\mathsf{H}(\mathbf{x}+\mathbf{y})| \ge |\mathsf{H}\mathbf{y}| - |\mathsf{H}\mathbf{x}| \ge \alpha \left(|\mathbf{y}| - \frac{|\mathsf{H}\mathbf{x}|}{\alpha}\right) \ge \alpha \left(|\mathbf{y}| -\frac{\gamma N}{4}\right)+1.
    \end{equation}
% \begin{proof}
%     Notice that for any bitstring obeying $|\mathbf{y}| \le \gamma N$: 
%     For $|\mathbf{y}|$ in the specified domain, (\ref{eq:weakerbarrier}) follows immediately.
% \end{proof}

Around any bitstring $\tilde{\bz}$ at low-energy $|\mathsf{H}\tilde{\bz}| \le \alpha\gamma N/4-1$, there is a subspace containing bitstrings of the form $\tilde{\bz}+\mathbf{y}$ with $\gamma N/2\le |\mathbf{y}|\le \gamma N$, such that all states saturating this inequality have a high number of flipped parity checks $\ge N_*$, where we now define \begin{equation}
     N_*:= \left\lfloor \frac{\alpha \gamma N}{4}\right\rfloor. \label{eq:Nstarnew}
\end{equation}
%In other words, any bitstring that is not in but close to the subspace with Hamming distance $\le \Delta'+1$ has energy $\ge N_*$, where we have defined $N_*$ differently than \eqref{eq:well<D}. 
We define projector \begin{equation}
    P_{\tilde{\bz}} = \sum_{\bs:|\bs-\tilde{\bz}|\le \gamma N/2, |\mathsf{H}\mathbf{s}| \le  N_*-1  } \ket{\bs} \bra{\bs},
\end{equation}
that projects onto a subspace labeled by the ``well" $\tilde{\bz}$ (for this subsection, we will end up replacing codeword with well).  Without loss of generality, we henceforth choose $\tilde{\bz}$ to be a configuration with as few parity checks flipped as possible, for each well.  For the region $|\bs-\tilde{\bz}|> \gamma N/2$ outside, we can find another low-energy $\tilde{\bz}'$ and define its corresponding well. Because starting at the bottom of well $\tilde{\bz}$, we know that all states a distance between $\gamma N/2$ and $\gamma N$ away from $\tilde{\bz}$ have a large number of flipped parity checks $\ge N_*$, clearly no two wells, which are restricted to states with at most $N_*-1$ flipped parity checks, will ever overlap.  Therefore, the wells  
%the Hamming distance is macroscopically larger than $\gamma N$: \begin{equation}
%     |\tilde{\bz}'-\tilde{\bz}|\ge \frac{\gamma N}{2}+ \mlr{\alpha\lr{\gamma N-\frac{\gamma N}{4}}-\lr{\frac{\alpha\gamma N}{4}-1}}\frac{1}{q}  > \frac{\gamma N}{2},
% \end{equation}
%so that the two wells do not overlap (and
cannot be connected by any $\Delta'$-local $V'$: $P_{\tilde{\bz}}V' P_{\tilde{\bz}'}=0$.

Repeating this process yields a set of wells $\tilde{\bz}$ that include \emph{all} low energy configurations, which are connected only through high energy configurations.  Any state obeying (\ref{eq:Hx<gamma}) belongs to a unique well with a macroscopic energy barrier. Define
\begin{equation}
    P_> := 1 - \sum_{\text{well }\tilde{\mathbf{z}}}P_{\tilde{\bz}},
\end{equation}
which projects onto the remaining bitstrings that all satisfy $|\mathsf{H}\bs|\ge N_*$, while (by the construction above) all bit strings with $|\mathsf{H}\bs|<N_*$ necessarily belong to one of the $P_{\tilde{\bz}}$. % Unlike for the c3LTC, $P_>$ might in principle project onto modes which violate $1+\frac{1}{3}\alpha\gamma N$ parity checks.

At this point, we now follow the proof of Theorem \ref{thm:main} verbatim, upon replacing $N_*$ with (\ref{eq:Nstarnew}), to deduce that all sufficiently low energy eigenstates are trapped very close to the bottom of a single well.  Notice that in some wells, the smallest value of $n$ that exists in the decomposition (\ref{eq:psi=psin}) may be close to $N_*$, but this does not actually change our arguments. The O(1) constants which have changed include $\mu$ (due to (\ref{eq:Nstarnew})); we also must take a larger value of $\delta$ to account for the fact that there are $\gg 2^K$ wells (but $\ll 2^N$); from (\ref{eq:leak}), shifting $\delta \rightarrow \delta + \frac{1}{2}$ suffices to account for this increased number of low energy wells, while maintaining the exponential localization of all low-energy eigenstates to a single well. In turn, these modifications lead to a more stringent bound on $\epsilon$ at which we can prove localization; of course, $\epsilon$ is still O(1). % Lastly, the discussion around Proposition \ref{prop:cn<} requires a trivial modification, since $P_>$ is no longer merely a projection onto states with $\ge N_*$ flipped parity checks, but rather $\ge \frac{1}{2}N_*$.

\end{document}